\documentclass[a4paper,12pt,USenglish]{article}
\usepackage{fullpage}
\usepackage{amsmath, amssymb}
\usepackage{enumerate}

\usepackage[nocompress]{cite}

\renewcommand{\angle}[1]{\langle{#1}\rangle}

\usepackage[dvipsnames,usenames]{xcolor}
\usepackage{graphicx}
\graphicspath{{figures/}} 
\usepackage{algorithmic}

\RequirePackage{hyperref}
\hypersetup{colorlinks=true, urlcolor=Blue, citecolor=Green, linkcolor=BrickRed, breaklinks, unicode}

\newcommand{\lef}{\ensuremath\mathsf{left}}
\newcommand{\rig}{\ensuremath\mathsf{right}}

\newcommand{\rmb}{\ensuremath\mathsf{rmb}}

\newcommand{\Sum}{\ensuremath\mathsf{sum}}
\newcommand{\Update}{\ensuremath\mathsf{update}}
\newcommand{\Access}{\ensuremath\mathsf{access}}
\newcommand{\Select}{\ensuremath\mathsf{select}}

\newtheorem{theorem}{Theorem}

\newtheorem{lemma}[theorem]{Lemma}
\newtheorem{corollary}[theorem]{Corollary}
\newenvironment{proof}{

\noindent{\bf Proof:}} {\hfill$\blacksquare$\medskip}

\title{Partial Sums on the Ultra-Wide Word RAM\footnote{An extended abstract appeared at the \emph{16th Theory and Applications of Models of Computation}\cite{BG2020}}}

\author{Philip Bille \\\texttt{phbi@dtu.dk} \and Inge Li G{\o}rtz \\\texttt{inge@dtu.dk} \and Frederik Rye Skjoldjensen \\\texttt{f.skjoldjensen@gmail.com}}
\date{}

\begin{document}
\maketitle

\begin{abstract}
	We consider the classic partial sums problem on the ultra-wide word RAM model of computation. This model extends the classic $w$-bit word RAM model with special ultrawords of length $w^2$ bits that support standard arithmetic and boolean operation and scattered memory access operations that can access $w$ (non-contiguous) locations in memory. The ultra-wide word RAM model captures (and idealizes) modern vector processor architectures.

	Our main result is a new in-place data structure for the partial sum problem that only stores a constant number of ultrawords in addition to the input and  supports operations in doubly logarithmic time. This matches the best known time bounds for the problem (among polynomial space data structures) while improving the space from superlinear to a constant number of ultrawords. 
	Our results are based on a simple and elegant in-place word RAM data structure, known as the Fenwick tree. Our main technical contribution is a new efficient parallel ultra-wide word RAM implementation of the Fenwick tree, which is likely of independent interest. 
\end{abstract}

\section{Introduction}
Let $A[1,\ldots, n]$ be an array of integers of length $n$. The \emph{partial sums problem} is to maintain a data structure for $A$ under the following operations: 

\begin{itemize}
    \item $\Sum(i)$: return $\sum_{k = 1}^i A[k]$.
    \item $\Update(i, \Delta)$: set $A[i] \leftarrow A[i] + \Delta$. 
\end{itemize}

The partial sums problem is a classic and well-studied data structure problem~\cite{Fredman1982,FS1989, Fenwick1994,HSS2011, HR2003,HRS1996, RRR2001,PD2006,Dietz1989, Yao1985, BCPS2017, BCCGSVV2018, HF1998, Fredman1981,FMS1997, BFK1981, BG2002, BG2001}. Partial sums is a natural range query problem with applications in areas such as list indexing and dynamic ranking~\cite{Dietz1989}, dynamic arrays \cite{RRR2001, BCCGSVV2018}, and arithmetic coding \cite{Fenwick1994,Ryabko1992}. From a lower bound perspective, the problem has been central in the development of new techniques for proving lower bounds~\cite{Miltersen1999}. In classic models of computation the complexity of the partial sums problem is well-understood with tight  logarithmic upper and lower bounds on the operations~\cite{PD2006}. Hence, a natural question is if practical models of computation capturing modern hardware advances will allow us the overcome the logarithmic barrier.

One such model is the \emph{RAM with byte overlap} (RAMBO) model of computation~\cite{FS1989, Brodnik1995, BCFKM2005}. The RAMBO model extends the standard $w$-bit word RAM model~\cite{Hagerup1998} with special words where individual bits are shared among other words, i.e., changing a bit in a word will also change the bit in the words that share that bit. The precise model depends on the layout of shared bits. This memory architecture is feasible to design in hardware and prototypes have been built~\cite{LMSTBK1999}. In the RAMBO model Brodnik et al.~\cite{BKMN2006} gave a time-space trade-off for partial sums that uses $O(n^{w/2^\tau} + n)$ space and supports operations in $O(\tau)$ time and for a parameter $\tau$, $1 \leq \tau \leq \log \log n$. Here, the $n$ term in the space bound is for the special words with shared bits (organized in a tree layout) and the $O(n^{w/2^\tau})$ term is for standard words. Plugging in constant $\tau$, this gives an $O(n^{\epsilon w} + n)$ space and constant time solution, for any $\epsilon > 0$. At the other extreme, with $\tau = \log \log n$, this gives an $O(n)$ space and $O(\log \log n)$ time solution.

More recently, Farzan et al.~\cite{FLNS2015} introduced the \emph{ultra-wide word RAM} (UWRAM) model of computation. The UWRAM model also extends the word RAM model, but with special \emph{ultrawords} of length $w^2$ bits. The model supports standard arithmetic and boolean operations on ultrawords and  \emph{scattered} memory access operations that access $w$ locations in memory specified by an ultraword in parallel. The UWRAM captures modern vector processor architectures~\cite{Stephensetal2017, Reinders2013,LNOM2008,CRDI2007}. We present the details of the UWRAM model in Section~\ref{sec:models}. Farzan et al.~\cite{FLNS2015} showed how to simulate algorithms on RAMBO model on the UWRAM model at the cost of slightly increasing space. Simulating the above solution for partial sums they gave a time-space trade-off for partial sums that uses $O(n^{w/2^\tau} + nw\log n)$ space and supports operations in $O(\tau)$ time and for a parameter $\tau$, $1 \leq \tau \leq \log \log n$.  For constant $\tau$, this is $O(n^{\epsilon w} + n w \log n)$ space and constant time, for any $\epsilon > 0$, and for $\tau = \log \log n$ this is $O(n w \log \log n)$ space and $O(\log \log n)$ time.

\subsection{Setup and Results} 

We revisit the partial sums problem on the UWRAM and present a simple new algorithm that significantly improves the space overhead of the previous solutions. Let $A$ be an array of $n$ $w$-bit integers. An \emph{in-place data structure} for the partial sums problem is a data structure that modifies the input array $A$, e.g., by replacing some of the entries in $A$, to efficiently support operations. In addition to the modified array the data structure is only allowed to store $O(1)$ of ultrawords. This definition extends the standard in-place/implicit data structure concept~\cite{Williams1964,MS1980,SS1987,FMP2007,CC2009} to the UWRAM, by allowing a constant number of ultrawords to be stored instead of (standard) words. Clearly, without this modification computation on ultrawords is impossible. As in Farzan et. al.~\cite{FLNS2015} we distinguish between the \emph{restricted UWRAM} that supports a minimal set of instructions on ultrawords consisting of addition, subtraction, shifts, and bitwise boolean operations and the \emph{multiplication UWRAM} that extends the instruction set of the restricted UWRAM with a multiplication operation on ultrawords. We show the following main result:  

\begin{theorem}\label{thm:main}
Given an array $A$ of $n$ $w$-bit integers, we can construct in-place partial sums data structures for $A$ that support $\Sum$ and $\Update$ operations in $O(\log \log n)$ time on a restricted UWRAM.   
\end{theorem} 
Compared to the previous result, Theorem~\ref{thm:main} matches the $O(\log \log n)$ time bound of Farzan et. al.~\cite{FLNS2015} (with parameter $\tau = \Theta(\log \log n)$ while improving the space overhead from $O(nw\log n)$ to a constant number of ultrawords. This is important in practical applications since modern vector processors have a very limited number of ultrawords available. 

Technically, our solution is based on a simple and elegant in-place word RAM data structure, called the  \emph{Fenwick tree} (see Section~\ref{sec:fenwicktree} for a detailed description). The Fenwick tree support operations in $O(\log n)$ by sequentially traversing an implicit tree structure. We show how to efficiently compute the access pattern on the tree structure in parallel using prefix sum computations on ultrawords. Then, given the locations to access we use scattered memory operations to access them all in parallel. In total, this leads to the exponential improvement of Fenwick trees. The main bottleneck in our algorithm is the prefix sum computation. Interestingly, if we allow multiplication we can compute prefix sums in constant time leading to the following Corollary for the multiplication UWRAM: 
\begin{corollary}\label{cor:multiplication}
	Given an array $A$ of $n$ $w$-bit integers, we can construct in-place partial sums data structures for $A$ that support $\Sum$ and $\Update$ operations in constant time on a multiplication UWRAM.   
\end{corollary} 
Multiplication (or prefix sum computation) is not an $\textsc{AC}^0$ operation (it cannot be implemented by a constant depth, polynomial size circuit) and therefore likely not practical to implement on ultraword. However, Corollary~\ref{cor:multiplication} shows that we can achieve significant improvements on the UWRAM with special operations. Since UWRAM capture modern processors, we believe it is worth investigating further, and that our work is a first step in this direction. 

\subsection{Outline}
The paper is organized as follows. In Section~\ref{sec:models} and~\ref{sec:fenwicktree} we review the UWRAM model of computation and the  Fenwick tree. In Section~\ref{sec:partialsumUWRAM} we present our UWRAM implementation of the Fenwick tree. Finally, in Section~\ref{sec:extensions} we discuss extensions of the result and open problems.

\section{The Ultra-Wide Word RAM Model}\label{sec:models}
The \emph{word RAM} model of computation~\cite{Hagerup1998} consists of an infinite memory of $w$-bit words and an instruction set of arithmetic, boolean, and memory access instructions such as the ones available in standard programming languages such as $C$. We assume that we can store a pointer into the input in a single word and  hence $w \geq \log n$, where $n$ is the size of the input. The time complexity of a word RAM algorithm is the number of instructions and the space complexity is the number of words used by the algorithm. 

The \emph{ultra-wide word RAM} (UWRAM) model of  computation~\cite{FLNS2015} extends the word RAM model with special ultrawords of $w^2$ bits. We distinguish between the \emph{restricted UWRAM} that supports a minimal set of instructions on ultrawords consisting of addition, subtraction, shifts, and bitwise boolean operations and the \emph{multiplication UWRAM} that additionally supports multiplication. The time complexity is the number of instruction (on standard words or ultrawords) and the space complexity is the number of (standard) words used by the algorithm. The restricted UWRAM captures modern vector processor architectures~\cite{Stephensetal2017, Reinders2013,LNOM2008,CRDI2007}. For instance, the Intel AVX-512 vector extension~\cite{Reinders2013} support similar operations on 512-bit wide words (i.e., a factor of $8$ compared to $64^2 = 4096$). 

\subsection{Word-Level Parallelism}
\begin{figure}[t]
\begin{center}
  \includegraphics[scale=0.5]{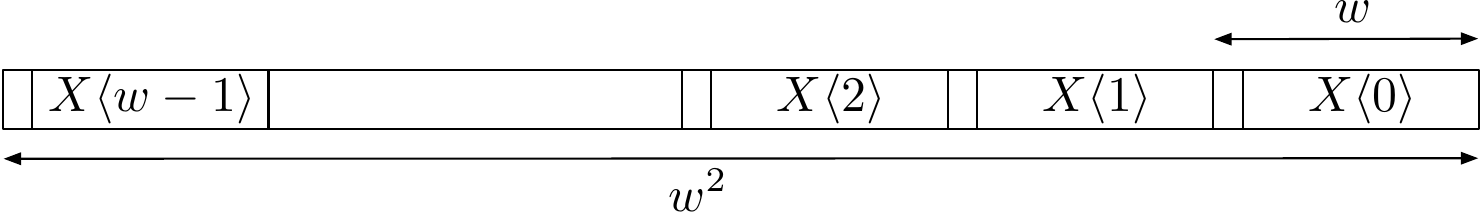}
  \caption{\label{fig:ultraword} The layout of an ultraword of $w^2$ divided into $w$ words each of $w$ bits. The leftmost bit of each word is reserved to be a test bit.}   
  \end{center}
\end{figure}

Due to their similarities, we can adopt many word-level parallelism techniques from the word RAM to the UWRAM.  We briefly review the key primitives and techniques that we will use. 

Let $X$ be an ultraword of $w^2$ bits. We often view $X$ as divided into $w$ words of $w$ consecutive bits each. See Figure~\ref{fig:ultraword}. We number the words in $X$ from right-to-left starting from $0$ and use the notation $X\angle{j}$ to denote the $j$th word in $X$. Similarly, the bits of each word $X\angle{j}$ are numbered from right-to-left starting from $0$. If only the rightmost $\ell \leq w$ words in $X$ are non-zero, we say that $X$ has \emph{length} $\ell$. For simplicity in the presentation, we reserve the leftmost bit of each word to be a \emph{test bit} for word-level parallelism operations. One may always remove this assumption at no asymptotic cost, e.g., by using two words in an ultraword to simulate each single word. 

We now show how to implement common operations on ultrawords that we will use later. Most of these are already available in hardware on modern vector processor architectures. Componentwise arithmetic and bitwise operation are straightforward to implement using standard word-level parallelism techniques from the word RAM . For instance, given ultrawords $X$ and $Y$, we can compute the componentwise addition, i.e., the ultraword $Z$ such that $Z\angle{j} = X\angle{j} + Y\angle{j}$ for $j =0, \ldots, w-1$ by adding $X$ and $Y$ and $\&$'ing with the mask $(01^{w-1})^w$ to clear any test bits (we use exponentiation to denote bit repetition, i.e., $0^31 = 0001$). We can also compare $X$ and $Y$ componentwise by $|$'ing in the test bits of $X$, subtracting $Y$, and masking out the test bits by $\&$'ing with $(10^{w-1})^w$. The $j$th test bit of the result contains a $1$ iff $X\angle{j} \geq Y\angle{j}$. Given $X$ and another ultraword $T$ containing only test bits, we can extract the words in $X$ according to the test bits, i.e., the ultraword $E$ such that $E\angle{j} = X\angle{j}$ if the $j$th test bit of $T$ is $1$ and $E\angle{j} = 0$ otherwise. To do so we copy the test bits by a subtracting $(0^{w-1}1)^w$ from $T$ and $\&$'ing the result with $X$. All of the above mentioned operation take constant time on a restricted UWRAM. Given an ultraword $X$ of length $\ell$, the \emph{prefix sum} of $X$ is the ultraword $P$ of length $\ell$, such that $P\angle{j} = \sum_{k \leq j} X\angle{k}$. We assume here that the integers computed in the prefix sum never exceed the maximum size available in a word such that $P\angle{j}$ is always well-defined. We need the following result. 

\begin{lemma}\label{lem:prefixsum}
Given an ultraword $X$ of length $\ell$ we can compute the prefix sum of $X$ in $O(\log \ell)$ time on a restricted UWRAM and in $O(1)$ time on a multiplication UWRAM. 
\end{lemma}
\begin{proof}
First consider the restricted UWRAM. We implement a standard parallel prefix-sum algorithm~\cite{LF1980} (see also the survey by Blelloch~\cite{Blelloch1990}). For simplicity, we assume that $\ell$ is a power of two. The algorithm consists of two phases that conceptually construct and traverse a perfectly balanced binary tree $T$ of height $\log \ell$ whose leaves are the $\ell$ words of $X$. 

Given an internal node $v$ in $T$, let $v_\lef$ and $v_\rig$ denote the left and right child of $v$, respectively. The first phase performs a bottom-up traversal of $T$ and computes for each node $v$ an integer $b(v)$. If $v$ is a leaf, $b(v)$ is the corresponding integer in $X$ and if $v$ is an internal node $b(v) = b(v_\lef) + b(v_\rig)$. The second phase performs a top-down traversal of $T$ and computes an integer $t(v)$. If $v$ is the root then $t(v) = 0$ and if $v$ is an internal node then $t(v_\lef) = t(v)$ and $t(v_\rig) = t(v_\lef) + b(v_\rig)$. After the second phase the integers at the leaves is the prefix sum shifted by a single element and missing the last element. We shift and add the last element to produce the final prefix sum. Since $T$ is perfectly balanced we can implement each level of a phase in constant time using shifting and addition. The final shift and addition of the last element takes constant time. It follow that the total time is $O(\log \ell)$. During the computation we only need to maintain all of the values in a constant number of ultrawords.

Next consider the multiplication instruction set. We can then simply multiply $X$ with the constant $(0^{w-1}1)^w$ and mask out the $\ell$ rightmost words of the result to produce the prefix sum. See Hagerup~\cite{Hagerup1998} for a detailed description of why this is correct. In total this uses $O(1)$ time.  

\end{proof}

\subsection{Memory Access}
The UWRAM supports standard memory access operation to read or write a single word or a sequence of $w$ contiguous words. More interestingly, the UWRAM also supports \emph{scattered} access operations that access $w$ memory locations (not necessarily contiguous) in parallel. Given an ultraword $A$ containing $w$ memory addresses, a \emph{scattered read} loads the contents of the addresses into an ultraword $X$, such that $X\angle{j}$ contains the contents of memory location $A\angle{j}$. Given two ultrawords $A$ and $X$ \emph{scattered write} sets the contents memory location $A\angle{j}$ to be $X\angle{j}$. Scattered memory accesses captures the memory model used by IBM's \emph{Cell} architecture~\cite{CRDI2007}. Scattered memory access operations were also proposed by Larsen and Pagh~\cite{LP2012} in the context of the I/O model of computation. 

\section{Fenwick Trees}\label{sec:fenwicktree}
\begin{figure}[t]
\begin{center}
  \includegraphics[scale=0.5]{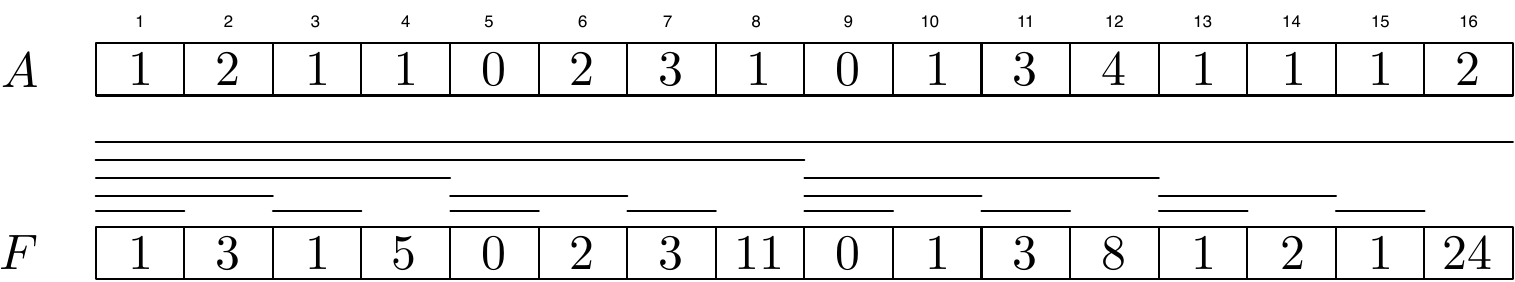}
  \caption{\label{fig:fenwicktree} A array $A$ and the Fenwick tree $F$. The lines above $F$ indicate the partial sum of $A$ stored at the rightmost endpoint of the line. For instance, the $F[12] = A[9] + A[10] + A[11] + A[12] = 0 + 1 + 3 + 4 = 8$.}   
  \end{center}
\end{figure}
Let $A$ be an array of $n$ $w$-bit integers and assume for simplicity that $n$ is a power of two. The Fenwick tree~\cite{Fenwick1994, Ryabko1992} is an in-place data structure that replaces the array $A$ as follows. If $n = 1$, then leave $A$ unchanged. Otherwise, replace all values at even entries  $A[2i]$ by the sum $A[2i - 1] + A[2i]$. Then, recurse on the subarray $A[2, 4, \ldots, n]$. The resulting array $F$ stores a subset of the partial sums of $A$ organized in a tree layout (see Figure~\ref{eq:sequences}). 

 To answer $\Sum(i)$ query, we compute a sequence of indices in $F$ and add the values in $F$ at these indices together. Let $\rmb(x)$ denote the position of the rightmost bit in an integer $x$. Define the \emph{sum sequence} $i^s_1, \ldots, i^s_r$ given by $i^s_1 = i$ and $i^s_j = i^s_{j-1} - 2^{\rmb(i^s_{j-1})}$, for $j = 2, \ldots, r$. The final element $i^s_r$ is $0$.   We compute and return $F[i^s_1] + F[i^s_2] + \cdots + F[i^s_{r-1}]$. For instance, for $i = 13= (1101)_2$ the sum sequence is $13, 12, 8, 0 = (1101)_2, (1100)_2, (1000)_2, (0000)_2$. Hence, $\Sum(13) = F[13] + F[12] + F[8] = 1 + 8 + 11 = 20 = A[1] + \cdots + A[13]$. We access at most $O(\log n)$ entries in $F$ and hence the total time for $\Sum$ is $O(\log n)$. Note that we can always recover the original array $A$ using the $\Sum$ operation, since $A[i] = \Sum(i) - \Sum(i-1)$.

To compute $\Update(i, \Delta)$, we compute a sequence of indices in $F$ and add $\Delta$ to the values in $F$ at each of these indices. Define the \emph{update sequence} $i^u_1, \ldots, i^u_t$ given by $i^u_1 = i$ and $i^u_j = i^u_{j-1} + 2^{\rmb(i^u_{j-1})}$, for $j = 2, \ldots, t$. The final element $i^u_t$ is $2n$. We set $F[i^u_1] = F[i^u_1] + \Delta, \ldots, F[i^u_t] = F[i^u_{t-1}] + \Delta$. For instance, for $i = 13$ the update sequence is $13, 14, 16, 32$. Hence, $\Update(13, 5)$ adds $5$ to $F[13]$, $F[14]$, and $F[16]$. Similar to the $\Sum$ operation, the total running time for $\Update$ is $O(\log n)$.

\section{Partial Sums on the Ultra-Wide Word RAM}\label{sec:partialsumUWRAM}
We now present an efficient implementation of Fenwick trees on the UWRAM model of computation. We only store the Fenwick tree, as the array $F$ described in Section~\ref{sec:fenwicktree} and a constant number of ultraword constants that we use for computation. We first show some basic properties of the sum and update sequences in Section~\ref{sec:sequences}, before presenting our UWRAM implementation of the operations in Sections~\ref{sec:sum} and~\ref{sec:update}. 

\subsection{Computing Sum and Update Sequences}\label{sec:sequences} 
To compute the sum and update sequences we cannot directly apply the recursive definitions, since this would need $\Omega(\log n)$ steps. Instead, we show how to express the sequences as a prefix sum that we can efficiently derive from the input integer $i$. Then, using Lemma~\ref{lem:prefixsum} we will show how to compute it in on the UWRAM in the following sections. 

 Let $i^s_1, \ldots, i^s_r$ and $i^u_1, \ldots, i^u_t$ be the sum sequence and update sequences, respectively, for $i$ as defined in Section~\ref{sec:fenwicktree}. Define the \emph{offset sum sequence} $o^s_{1}, \ldots, o^s_{r-1}$ and \emph{offset update sequence} $o^u_{1}, \ldots, o^u_{t-1}$ for $i$ to be the sequences of differences of the sum and update sequences, respectively, that is, $o^s_{j} = i^s_{j+1} - i^s_{j}$, for $j = 1, \ldots, r-1$ and $o^u_{j} = i^u_{j+1} - i^u_{j}$, for $j = 1, \ldots, t-1$. By definition, we have that 
 \begin{equation}\label{eq:sequences}
i^s_j = i + \left(\sum_{k < j} o^s_k\right)  \qquad \qquad
i^u_j = i + \left(\sum_{k < j} o^u_k\right) 
\end{equation}

We also have that $o^s_j = -2^{\rmb(i^s_j)}$ and hence each sum offset is a power of $2$ corresponding to the rightmost $1$ bit in $i^s_j$. Thus, $o^s_1$ corresponds to the rightmost $1$ in $i^s_1 = i$. Adding $o^s_1 = -2^{\rmb(i)}$ (i.e., subtracting $2^{\rmb(i)}$) "clears" the rightmost $1$ bit in $i$. Thus, $o^s_2$ corresponds to the $1$ bit in $i$ immediately to left of the rightmost $1$ bit. In general, we have that $o^s_j = -2^b$, where $b$ is the position of the $j$th rightmost bit in $i$, for $j=1, \ldots, r-1$. For instance, for $i = 13 = (1101)_2$ the offset sum sequence is $-1, -4, -8$ corresponding to the three $1$ bits in the binary representation of $i$. 

Similarly, for the update offsets, we have that $o^u_j = 2^{\rmb(i^u_j)}$. Hence, $o^u_1$ corresponds to rightmost $1$ in $i$. Adding $o^u_1 = 2^{\rmb(i^u_1)}$ clears the  rightmost consecutive group of $1$ bits in $i$ and flips the following  $0$ bit to $1$. In general, we have that $o^u_j = 2^b$, where $b$ is the position of the $j$th rightmost $0$ to the left of $\rmb(i)$, for $j=2, \ldots, t-1$. For instance, for $i = 13 = (01101)_2$ the offset update sequence is $1, 2, 16$. 

\subsection{Sum}\label{sec:sum}
To compute the $\Sum(i)$, the main idea is to first construct the sum sequence in an ultraword, then use a scattered read to retrieve the entries from $F$ in parallel into another ultraword, and finally sum the entries of this ultraword to compute the final result. We do this in $3$ steps as follows. See Figure~\ref{fig:sumexample} for an example of the computed ultrawords during the algorithm.

\begin{figure}[t]
\begin{center}
  \includegraphics[scale=0.5]{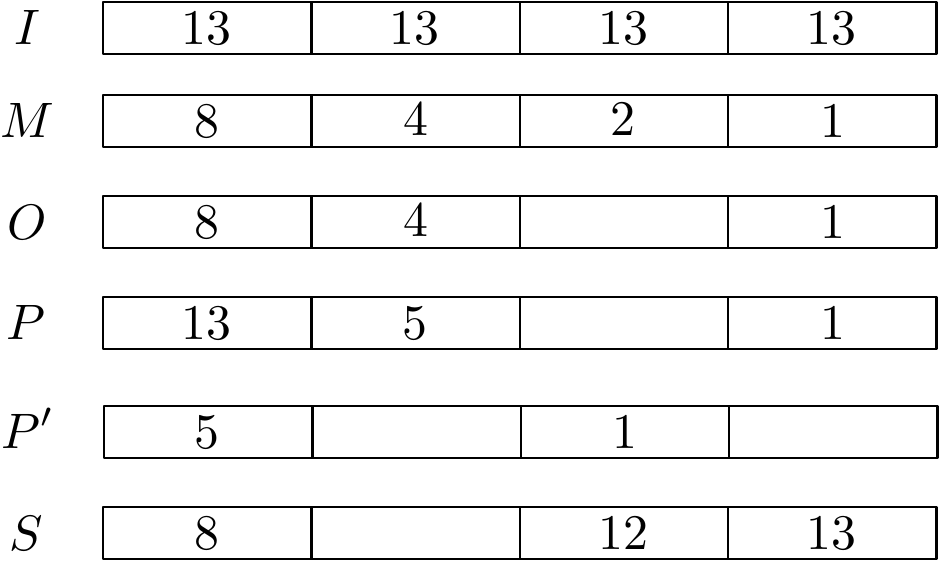}
  \caption{\label{fig:sumexample} Computing the sum sequence for $i = 13 = (1011)_2$. Words with $0$ are left blank. $I$ contains duplicates of $i$. $M$ is a precomputed mask. $O$ is the bitwise $\&$ of $I$ and $M$. $P$ is the prefix sum of the non-zero words in $O$. $P'$ is $P$ shifted left by one word. $S$ is the sum sequence obtained by componentwise subtraction of $P'$ from $I$.}	
\end{center}
\end{figure}

\paragraph{Step 1: Compute Offsets} 
Compute the ultraword $O$ such that $O\angle{j} = 2^j$ if $-2^j$ is an offset for $i$ and $0$ otherwise, i.e., the non-zero entries of $O$ is the offset sequence for $i$. To do so we first construct the ultraword $I$ consisting of $\log n$ duplicates of $i$, i.e., $I\angle{j} = i$ for $j = 1,\ldots, \log n$. We then compute the bitwise $\&$ of $I$ and a mask $M$, such that $M\angle{j} = 2^j$ for $j = 1,\ldots, \log n$, i.e., bit $j$ of $M\angle{j} = 1$ and the other bits of $M\angle{j}$ are $0$. By the discussion in Section~\ref{sec:sequences} the resulting ultraword is $O$.

On the multiplication UWRAM we can construct $I$ in constant time by multiplying $i$ with $(0^{w-1}1)^w$. On the restricted UWRAM we can construct $I$ in $O(\log \log n)$ time by repeatedly doubling using shifts and bitwise $|$. The rest of the computation takes constant time in both models. 


\paragraph{Step 2: Compute Sum Sequence} 
Compute an ultraword $S$ of length $\log n$ whose non-zero entries is the sum sequence $i^s_1, \ldots, i^s_{r-1}$. To do so we first compute the prefix sum $P$ of the non-zero words of $O$, i.e., we compute the prefix sum of $O$ and then extract the words corresponding to non-zero words in $O$. Then we shift $P$ by $1$ word to the left to produce an ultraword $P'$ and finally subtract $P'$ from $I$ to produce an ultraword $S$. By \eqref{eq:sequences} the non-zero words in $S$ is the sum sequence for $i$.

By Lemma~\ref{lem:prefixsum} the prefix sum computation takes constant time on a multiplication UWRAM and $O(\log \log n)$ time on a restricted UWRAM. The remaining steps take constant time. 

\paragraph{Step 3: Compute Sum}
Finally, we compute $F[i^s_1] + F[i^s_2] + \cdots + F[i^s_{r-1}]$. To do so we do a scattered read on $S$ to retrieve $F[i^u_1], \ldots, F[i^u_{s-1}]$ into a single ultraword $F'$ and compute a prefix sum on $F'$. The sum is then the last word in the result. The scattered read takes constant time. The prefix sum computation takes constant time on a multiplication UWRAM and $O(\log \log n)$ time on a restricted UWRAM. We assume here that $F[0] = 0$. If not we may simply temporarily set $F[0] = 0$ during the computation. Also note that it suffices to perform the first phase of the prefix sum computation as discussed in the proof of Lemma~\ref{lem:prefixsum} since we only need the sum of all of the retrieved entries. 

\medskip

In total, the $\Sum$ operation takes constant time on a multiplication UWRAM and $O(\log \log n)$ time on a restricted UWRAM. 

\subsection{Update}\label{sec:update}
We compute $\Update(i, \Delta)$ similar to our algorithm for $\Sum$. We describe how to modify each step of $\Sum$.

In step 1, we modify the computation of the ultraword $O$ such that it now contains the update offsets, that is, $O\angle{j} = 2^j$ if $2^j$ is an update offset for $i$ and $0$ otherwise. To do so we now construct a mask $M$ such that $M\angle{j}$ contains a $0$ in bit $j$ if $j$ is to the left of $\rmb(i)$ and $1$ elsewhere. We then compute a bitwise $|$ of $M$ and $I$ and negate the result. Finally, we set word $\rmb(i)$ of the result to be $2^{\rmb(i)}$. By the discussion in Section~\ref{eq:sequences} the resulting ultraword is $O$. 

In step 2, since $O$ now contains the offsets and not the negative offsets, we change the final subtraction to an addition to produce the update sequence stored in a single ultraword $U$. 

In step 3, we do a scattered read on $U$ to retrieve $F[i^u_1], \ldots, F[i^u_{s-1}]$ into a single ultraword $F'$. We then duplicate $\Delta$ to all words in an ultraword $D$ and add $D$ to $F'$ to produce an ultraword $F''$. Finally, we do a scattered write on $U$ and $F''$ to update $F$.  

The changes are straightforward to implement in the same time as above. Hence, the $\Update$ operation takes constant time on a multiplication UWRAM and $O(\log \log n)$ time on a restricted UWRAM. 

\medskip

In summary, we use $O(\log \log n)$ time on a restricted UWRAM and $O(1)$ time on a multiplication UWRAM for both operation. We only store the Fenwick tree in the array $F$ and a constant number of ultrawords. This completes the proof of Theorem~\ref{thm:main} and Corollary~\ref{cor:multiplication}.

\subsection{Extensions and Open Problems}\label{sec:extensions}
We sometimes also consider the following operations in the context of partial sums: 
\begin{itemize}
	\item $\Access(i)$: return $A[i]$. 
	\item $\Select(j)$: return the smallest $i$ such that $\Sum(i) \geq j$
\end{itemize}
As mentioned $\Access$ is trivial to support since $A(i) = \Sum(i) - \Sum(i-1)$. In contrast, the $\Select$ operation do not seem to easily lend itself to an efficient parallel implementation on the UWRAM. While it is straightforward to implement in $O(\log n)$ time by "top-down" traversal of the Fenwick tree our techniques do not appear be useful to speed up this solution on the UWRAM. We leave it as an open problem to investigate the complexity of the $\Select$ operation on the UWRAM. 

Our results leave the precise relation between UWRAM and RAMBO model of computation open. While Farzan et al.~\cite{FLNS2015} show how to simulate RAMBO algorithms  with a small overhead in space our results show that a direct approach to designing UWRAM algorithms can produce significantly better results. We wonder what the precise relation between the models are and if stronger simulation results are possible. 


\bibliographystyle{abbrv}
\bibliography{paper}

\end{document}